\documentclass[12pt,oneside]{amsart}
\usepackage[margin=1in]{geometry}
\usepackage{amsmath,amsthm,amssymb,amsfonts}
\usepackage{latexsym}
\usepackage[T2A]{fontenc}
\usepackage[utf8]{inputenc}
\usepackage{booktabs}
\usepackage{longtable}
\usepackage[font=small,labelfont=bf]{caption}
\usepackage[russian,english]{babel}
\usepackage{natbib}
\usepackage{ragged2e}
\usepackage{setspace}
\usepackage{xcolor}
\usepackage{csquotes}
\usepackage{multirow}
\newtheorem*{proposition*}{Proposition}
\newtheorem*{corollary*}{Corollary}

\newtheorem{proposition}{Proposition}
\newtheorem{corollary}{Corollary}
\newtheorem{lemma}{Lemma}
\newtheorem*{remark*}{Remark}
\newtheorem{definition}{Definition}
 \newtheorem{theorem}{Theorem}
 \newtheorem{example}{Example}
\newtheorem{claim}{Claim}

\usepackage[foot]{amsaddr}

\makeatletter
\renewcommand{\email}[2][]{%
  \ifx\emails\@empty\relax\else{\g@addto@macro\emails{,\space}}\fi%
  \@ifnotempty{#1}{\g@addto@macro\emails{\textrm{(#1)}\space}}%
  \g@addto@macro\emails{#2}%
}
\makeatother

\begin{document}
\onehalfspacing

\thispagestyle{empty}

\title[Immunity to Strategic Admissions]{Comparing School Choice and College Admission Mechanisms By Their Immunity to Strategic Admissions}

\author{Somouaoga Bonkoungou}
\address{Higher School of Economics, St.Petersburg}
\email{bkgsom@gmail.com}

\author{Alexander Nesterov}
\email{nesterovu@gmail.com}

\thanks{We thank Parag Pathak and Tayfun S\"onmez for their comments and continuous communication. We also thank Anna Bogomolnaia, Herve Moulin, Lars Ehlers, Sean Horan, Onur Kesten, Rustamdjan Hakimov, Szilvia Papai, Yan Chen, Fuhito Kojima, Alvin Roth, Fedor Sandomirskiy and Mikhail Panov for their feedback. Support from the Basic Research Program of the National Research University Higher School of Economics is gratefully acknowledged. The research is partially financed by the RFBR grant 20-01-00687.}

\date{January 16, 2019}

\begin{abstract}

Recently dozens of school districts and college admissions systems around the world have reformed their admission rules. As a main motivation for these reforms the policymakers cited strategic flaws of the rules: students had strong incentives to game the system, which caused dramatic consequences for non-strategic students. However, almost none of the new rules were strategy-proof. We explain this puzzle. We show that after the reforms the rules became more immune to strategic admissions: each student received a smaller set of schools that he can get in using a strategy, weakening incentives to manipulate. Simultaneously, the admission to each school became strategy-proof to a larger set of students, making the schools more available for non-strategic students. We also show that the existing explanation of the puzzle due to \cite{pathak2013} is incomplete.\\ 

%
%
%
\end{abstract}

\maketitle

\thispagestyle{empty}

\section{Introduction}
In recent years, dozens of school districts around the world have reformed their school admissions systems. Examples include education policy reforms for K-9 Boston Public Schools (BPS) in 2005, Chicago Selective High Schools (SHS) in 2009 and 2010, Denver Public Schools in 2012, Seattle Public Schools in 1999, Ghanaian Secondary Public Schools in 2007 and Primary Public Schools in more than 50 cities and provinces in England and Wales in 2005-2011. Like school admissions, many college admissions systems have also been reformed; well-known examples include college admissions in China and Taiwan. 

Sometimes the reforms were a pressing issue. The Chicago SHS, for example, called for a reform in a midstream of their admissions process. What were the policymakers concerned about, and what was it at stake for such a sudden midstream change? There are signs indicating that the public and the policymakers were concerned about the high vulnerability of the admissions mechanisms to strategic manipulations. For example, the former superintendent of the Boston BPS said that their mechanism \textit{\enquote{should be replaced with an alternative [...] that removes the incentives to game the system}} \citep{pathak2008}.

Indeed, this high vulnerability made strategy an essential decision for students and led to serious mismatches between schools and students. Strategic rankings were playing an unbalanced role in admissions versus priorities/grades, and that was perceived undesirable. For example, the education secretary in England remarked that their admissions system were forcing \textit{“many parents to play an admissions game with their children’s future}” \citep{pathak2013}. Prior to the reform in China, \textit{“a good score in the college entrance exam is worth less than a good strategy in the ranking of colleges''}  \citep{chenandkesten2017a,nie2007game}. 

These issues also compromise the perceived fairness of the system as the consequences to truthful students can be disastrous. The Chicago SHS called for reform after they observed that \textit{``High-scoring kids were being rejected simply because of the order in which they listed their college prep preferences''} \citep{pathak2013}. As one parent in China reports: \textit{``My child has been among the best students in his school and school district. [...] Unfortunately, he was not accepted by his first choice. After his first choice rejected him, his second and third choices were already full. My child had no choice but to repeat his senior year.''} \citep{chenandkesten2017a,nie2007game}. Reportedly, similar concerns resulted in protests in Taiwan \citep{dur2018}.

Did the reforms make the admissions mechanisms fully immune to manipulation? The answer is no. Except for Boston, each reform replaced one vulnerable mechanism with another vulnerable mechanism. This is a puzzle, given what motivated these reforms. But could it be that the new mechanisms are more immune to manipulations than the old ones? Did the reforms weakened the incentives to manipulate and made the consequences thereof less harmful to others? To address this question we develop a criterion to rank mechanisms by their level of what we call immunity to strategic admissions.  

To explain this criterion, let us begin with the most immune mechanisms. A mechanism is \textit{strategy-proof} when no student can ever gain from manipulating his preferences. That is, for such a mechanism the admission of each student $i$ to each school $ s$ is strategy-proof. Let us generalize this definition to any, possibly not strategy-proof, mechanism. We say that the \textbf{admission to school $ s$ is strategy-proof to student $ i$ via mechanism A} if none of $i$'s profitable manipulations gives him an admission at school $ s$.  In other words, all the misreports that result in student $ i$'s admission to school $ s$ via mechanism A are not profitable: in each of those instances, he is weakly better-off reporting his preferences truthfully. This student has no reason to misreport his preferences to mechanism $ A$ when aiming at school $ s$. Student $i$ may still profitably manipulate mechanism $A$ and get an admission at other schools --- but not at $ s$. 

We measure the level of immunity of a mechanism by how strategy-proof admission to each school is. Formally, mechanism $A$ is \textbf{more immune (to strategic admissions) than $B$} if for each student $i$ the set of schools whose admission is strategy-proof to $i$ via $B$ is a subset of the set of schools whose admission is strategy-proof to $i$ via $A$, while the converse is not true. Thus, with a more immune mechanism, each student faces more schools that he cannot be admitted to via a profitable strategy.

We find that each reform made the mechanisms more immune. For each student we compare by inclusion the two sets --- before and after the reforms --- of schools whose admission is strategy-proof. Each of the reforms enlarged this set. Simultaneously, following the reforms, the admission to each school became strategy-proof to a larger set of students.

Roughly, the reforms made the mechanism more immune by using one or both of the following features: they allowed students to submit longer lists of acceptable schools and made admission to every school less sensitive to its rank in the list. Intuitively, a longer list allows students to be less strategic about selecting which schools to include in the list, while lower sensitivity about ranking facilitates truthful ranking of the selected schools. 

We illustrate our concept and the result in the following example.

\begin{table}
\begin{tabular}{lllll}
\# & School Name  & US rank & min score & max score \\ \hline
1  & Payton & 9 & 898 & 900 \\
2  & Northside& 23& 894 & 900 \\
3  & Lane   & 69& 873 & 900 \\
4  & Young  & 71& 883 & 900 \\
5  & Jones  & 91& 891 & 900 \\ \hline
6  & Brooks & 186 & 799 & 890 \\
7  & Lindblom & 272 & 772 & 858 \\
8  & Westinghouse & 574 & 773 & 884 \\
9  & King   & 1,133   & 678 & 844 \\
10 & South Shore  & 6,066   & 684 & 820 \\ \hline
\end{tabular}
\caption{Chicago selective high schools (SHS): rankings and cutoff grades.}\label{tab-chicago}
\justifying{\footnotesize\linespread{1.0}{{\textit{Notes}: South Shore was added to SHS in 2011, followed by Hancock school in 2015. The US rank is according to US News \& World 2019 Report; the grades are reported for general admission in 2019, the maximal feasible grade is 900; the data source is www.go.cps.edu.}}}
\end{table}

\subsection*{Illustrative Example} Let us consider the two reforms of the Chicago SHS in 2009 and 2010. Each school uses a common priority based on students' composite scores. The admission to each of these schools is very competitive. To give you an idea, only 4 000 from more than 10 000 participants were admitted in the 2018 admission session. In 2009, the Chicago SHS replaced the Boston mechanism where students can rank only 4 schools ($\beta^4$) with a serial dictatorship with the same ranking constraint ($SD^4$).\footnote{The serial dictatorship is a mechanism where students follow the common priority order and choose their most preferred schools among those that remain. The definitions of the mechanisms are given in the next section. Constrained matching mechanisms were first studied in \cite{haeringer2009constrained} and \cite{calsamiglia2010constrained}.} In 2010 the mechanism stayed the same but the constrained was increased to 6 ($SD^6$). 

Among the 10 schools, 5 are elite, being the top 5 schools in the state of Illinois and among top 100 in the US (see Table \ref{tab-chicago}). These schools are preferred by most, if not all, students over each other school. For simplicity, let us suppose that students have tier preferences: each student prefers each top 5 school over each non top 5 school, but may differ on ranking schools in each tier. Let each school have 400 seats. 

Under the mechanism $ \beta^4$, each of the 400-highest priority students is guaranteed a seat in his most preferred school, while each other student may potentially get each school by a profitable manipulation. The admission to every school is thus strategy-proof only for the 400-highest priority students via $ \beta^4$.

However, under the mechanism $SD^4$, each of the 1600-highest priority students is guaranteed one of his 4 most preferred schools (Lemma \ref{lemma1}). The admission to every school is strategy-proof for each of them, while each other student can potentially get every school by a profitable manipulation. Under $SD^6$, this is the 2400-highest priority students; while only the 2000-highest priority students can be admitted to the top 5 schools. Following the reforms, the immunity of the mechanism increased significantly as the share of students for whom the admission to the top 5 schools is strategy-proof increased from 4\% to 24\% in 2009 and further to 100\% in 2010. End of the example.

Our results rationalize all the reforms in the light of each of the cited concerns. The reforms decreased each student's incentives to manipulate as measured by the range of schools he or she could get by manipulation. Intuitively, this range decreases due to the more competitive schools as it was in the Chicago example, which makes the incentives to manipulate even weaker. At the same time, the remaining manipulations harmed the truthful students less as the admission to each school became more strategy-proof. Again, in the Chicago example all high-scoring kids were safe to submit their preferences truthfully.

But the reforms could also be rationalized by that the instances without manipulations became more frequent. The state-of-the-art notion is \textit{manipulability} due to \cite{pathak2013}:  mechanism $A$ is less manipulable than mechanism $B$ if A is manipulable by at least one student in a subset of preference profiles where $B$ is manipulable. 

We found that this notion has limited applications for important reforms. First, contrary to what was claimed, it only partially explains a major reform in England and Wales that was followed by more than 50 local cities (see section 4 and Table \ref{reforms}). 

Second, under realistic assumption on students preferences, the notion is not satisfactory for rationalizing the 2009 reform of the Chicago SHS. Indeed, for each preference profile, it suffices that one student has a profitable manipulation to declare the profile as manipulable. However, with tier preference structure in the Chicago SHS --- strong competition for the elite schools --- at least one such student always exists. By truthfully ranking the top 5 schools the mechanisms that allow students to submit at most 5 schools will leave the seats at the other schools unassigned. Therefore, many students will be unmatched and at least one student will have a profitable manipulation. Before and after the 2009 reform, both mechanisms are manipulable at any tier preference profile. More generally, we show that constrained $SD$ is not manipulable if and only if the constraint is not binding (Proposition \ref{pro3}). However, immunity to strategic admissions changed significantly after both the 2009 and 2010 reform, and the fact that admission to the elite schools is now strategy-proof may rationalize why $SD^6$ has been used in Chicago ever since. 

Very recently \cite{decerf2018manipulability} used another frequency notion stemming from \cite{ar2016}: they compare mechanisms by inclusions of vulnerable individual preference relations, i.e., the preference relations for which truth is not always a dominant strategy. This notion explains reforms where the constrained list of the Boston mechanism were replaced by a constrained list of Gale-Shapley mechanism and where the list in Gale-Shapley were extended. The major difference with our notion is that we quantitatively measure what students can get or cannot get by manipulation. We particularly measure how schools are protected from strategic admissions.

The paper is organized as follows. Section 2 presents the model and the main definitions. Section 3 presents the main results for immunity to strategic admissions and section 4 compares them to the results for manipulability. Section 5 develops an equilibrium refinement of immunity to strategic admissions. Section \ref{conclusions} concludes.
\vspace{0.7cm}
\section{Model}
\vspace{0.7cm}
The school choice model originated in \cite{balinski1999} and \cite{abdul2003}; constrained school choice was first studied in \cite{haeringer2009constrained}.

There is a finite set $ I$ of students with a generic element $ i$ and a finite set $ S$ of schools with a generic element $ s$. Each student $ i$ has a strict preference relation $ P_{i}$ over $ S\cup \{\emptyset\}$ (where $ \emptyset$ stands for being unmatched).\footnote{A strict preference relation is a complete, transitive and asymmetric binary relation.} Each school $ s$ has a strict priority order $ \succ_{s}$ over the set $ I$ of students and a capacity $ q_{s}$ (a natural number indicating the number of seats available at school $ s$). For each student $ i$, let $ R_{i}$ denote the ``at least as good as'' relation associated with $ P_{i}$.\footnote{That is, for each $s,s'\in S\cup\{\emptyset\} $, $ s\mathrel{R_{i}}s'$ if and only $ s\mathrel{P_{i}}s'$ or $ s=s'$.} The list $ P=(P_{i})_{i\in I}$ is a preference profile, $ \succ=(\succ_{s})_{s\in S}$ is a priority profile and $ q=(q_{s})_{s\in S}$ is a capacity vector. We often write a preference profile $ P=(P_{i},P_{-i})$ to emphasize the preference relation of student $ i$.\footnote{More generally, we write a preference profile $ P=(P_{I'},P_{-I'})$ to emphasize the components of a subset $I' $ of students.} The tuple $ (I,S,P,\succ,q)$ is a school choice problem. We assume that there are more students than schools and at least two schools to reflect real-life school choice context.\footnote{That is, $ \lvert I\lvert > \lvert S\lvert \geq 2$.} We fix the set of students and the set of schools throughout the paper. For short, we call the pair $(\succ,q)$ a school choice \textbf{environment} and the triple $(P,\succ,q) $ a school choice \textbf{problem}.\footnote{We assume that schools are not strategic. In practice, the priorities are determined by law or by students' performances, and known to students before they submit their preferences.} School $ s$ is \textbf{acceptable} to student $ i$ if $ s \mathrel{P_{i}} \emptyset$. Otherwise, it is unacceptable. We will often specify a preference relation and a priority order as follows \begin{center}
\begin{tabular}{c| c}
  $ P_{i}$ & $ \succ_{s}$\\ \hline
  $s$ & $ i$\\
  $s'$ & $ j$\\
  $ \emptyset$ & $ k$
\end{tabular}
\end{center}
to indicate that student $ i$ prefers school $ s$ to school $ s'$ and finds no other school acceptable; and, when there are three students, that student $ i$ has the highest priority at school $ s$, student $ j$ next and student $ k$ last. 

A matching $ \mu$ is a function $ \mu:I \rightarrow S\cup \{\emptyset\}$ such that no school is assigned more students than its capacity.\footnote{That is for each school $ s$, $\lvert \mu^{-1}(s)\lvert \leq q_{s}$.} Let $ (P,\succ,q)$ be a problem. A matching $ \mu$ is \textbf{individually rational} under $ P$ if for each student $ i$, $ \mu(i) \mathrel{R_{i}} \emptyset$. Student $ i$ has \textbf{justified envy} over student $ j$ in the matching $ \mu$ if he prefers the school assigned to student $ j$ to his assignment and has higher priority than $ j$ at that school.\footnote{That is, for some school $ s$, $ \mu(j)=s$, $ s\mathrel{P_{i}} \mu(i)$ and $ i\mathrel{\succ_{s}}j$.} A matching $ \mu$ is \textbf{non-wasteful} if no student prefers a school that has an empty seat.\footnote{That is, there is no student $ i$ and a school $ s$ such that $ \lvert \mu^{-1}(s)\lvert < q_{s}$ and $ s \mathrel{P_{i}} \mu(i)$.} A matching is \textbf{stable} if 
\begin{itemize}
\item it is individually rational,
\item no student has a justified envy over another, and
\item it is non-wasteful.
\end{itemize}
A \textbf{mechanism} $ \varphi$ is a function that  maps school choice problems $ (P,\succ,q)$ to matchings. Let $ \varphi_{i}(P,\succ,q)$ denote the outcome for student $ i$. 
We present real-life mechanisms next.

\subsection{Mechanisms} Most real-life mechanisms can be described using the Gale-Shapley deferred acceptance mechanism.
\subsubsection*{Gale-Shapley}
\cite{galeandshapley} showed that for each problem, there is a stable matching. In addition, there is a student-optimal stable matching, that each student finds at least as good as any other stable matching.\footnote{\cite{galeandshapley} described an algorithm for finding this matching.} For each problem $ (P,\succ,q)$, let $ GS(P,\succ,q)$ denote the student-optimal stable matching.
\subsubsection*{Serial Dictatorship} In environments where schools have the same priority order, we abuse language and call the Gale-Shapley mechanism, serial dictatorship.\footnote{According to our definition, a mechanism has as a domain the set of all problems --- including problems where schools have different priorities.} Let $ SD(P,\succ,q)$ denote the matching assigned by this mechanism to the problem $ (P,\succ,q)$ where schools have the same priority under $ \succ$.
\subsubsection*{First-Preference-First} The set of schools are partitioned into \textbf{equal preference schools} and \textbf{first-preference-first} schools. For each problem $ (P,\succ,q)$, the mechanism assigns the matching $ GS(P,\succ',q)$ where the priority profile $ \succ'$ is obtained as follows:
\par 1.~for each equal preference school $ s$, $ \succ'_{s}=\,\succ_{s}$ and
\par 2.~for each first-preference-first school $ s$, $ \succ'_{s}$ is an adjustment of $ \succ_{s}$ with respect to $ P$:
\begin{itemize}
\item students who rank school $ s$ at a given position have higher priority under $ \succ'_{s}$ than students who rank it lower than this position, and
\item students who ranked school $ s$ at the same position are ordered according to $ \succ_{s}$.
\end{itemize}
Let $ FPF(P,\succ,q)$ denote the matching assigned by this mechanism to $ (P,\succ,q)$. Briefly, the original priority of each equal preference school remains unchanged, while the original priority of each first-preference-first school is adjusted by favoring higher ranks.  
\subsubsection*{Boston} The Boston mechanism is a first-preference-first mechanism where every school is a first-preference-first school. Let $ \beta(P,\succ,q)$ denote the matching assigned to $ (P,\succ,q)$.
\subsubsection*{Constrained versions} In practice, students are allowed to report a limited number of schools. This means that schools that are listed below a certain position are not considered. For each student $ i$, each preference relation $ P_{i}$ and each natural number $ k\leq \lvert S\lvert$, let $ P^{k}_{i}$ denote the truncation of $ P_{i}$ after the $ k$'th acceptable choice (if any). That is, every school ranked below the $ k$'th position under $ P_{i}$ is unacceptable under $ P^{k}_{i}$; otherwise, the ranking is as in $ P_{i}$. Let $ P^{k}=(P^{k}_{i})_{i\in I}$. The $ k$-constrained version $ \varphi^{k}$ of the mechanism $ \varphi$ is the mechanism that assigns to each problem $ (P,\succ,q)$ the matching $ \varphi(P^{k},\succ,q)$.
\subsubsection*{Chinese parallel}
This mechanism is determined by a parameter $ e\geq 1$ (a natural number).
For each problem $(P,\succ,q) $, the outcome is a sequential application of constrained $GS$. In the first round, students are matched according to $GS^{e}$. The matching is final for matched students, while unmatched students proceed to the next round. In the next round, each school reduces its capacity by the number of students assigned to it the last round and the unmatched students are matched according to $GS^{2e}$, and so on. Let $ Ch^{(e)}(P,\succ,q)$ denote the matching assigned by the mechanism to $ (P,\succ,q)$.\footnote{This definition of the Chinese parallel mechanisms is given only for the symmetric version where each round has the same length $e$. Our results hold also for the asymmetric mechanisms with different lengths of the rounds. See \cite{chenandkesten2017a} for details.}
\vspace{0.7cm}
\section{Results} \vspace{0.7cm}

We first introduce our immunity notion. We would like to test how immune is a mechanism to strategic admissions, focusing on individual students. We build a definition from the following question. When can we say that an admission to school $ s$ by student $i $ is not due to strategic manipulation? Obviously, when no profitable manipulation by $ i$ gives him an admission at $ s$, then none of this admissions is due to strategy. 
\vspace{0.5cm}
\begin{definition}\label{def-str-pr-adm} Let $ \varphi$ be a mechanism and $ (\succ,q)$ an environment. 
\begin{itemize} 
    \item The \textbf{admission to school $ s$ is strategy-proof to student $ i$ via mechanism $ \varphi$} if there is no preference profile $ P$ and a profitable manipulation $ P'_{i}$ that gives him $s$: 
\begin{equation*}
s=\varphi_{i}(P'_{i},P_{-i},\succ,q) \mathrel{P_{i}}\varphi_{i}(P,\succ,q).
\end{equation*} 
\item  
 Mechanism $ \varphi$ is \textbf{strategy-proof to student $ i$} if the admission to every school is strategy-proof to student $i$ via $ \varphi$.
\end{itemize}
\end{definition}
\vspace{0.5cm}

 We qualify these admissions as non-strategic and the rest as strategic. A closer definition in the literature is strategy-proofness saying that mechanism $ \varphi$ is strategy-proof if for each environment $ (\succ,q)$ and each student $ i$, there is no preference profile $ P$ and $ P'_{i}$ such that 
\begin{equation*}
    \varphi_{i}(P'_{i},P_{-i},\succ,q) \mathrel{P_{i}} \varphi_{i}(P,\succ,q).
\end{equation*}
If a mechanism $ \varphi$ is strategy-proof, then for each environment, the admission to every school is strategy-proof to every student via $ \varphi$. The unconstrained $ GS$, for example, is strategy-proof \citep{roth1982a,dubins}. Next, for each student, we find the set of schools for which the admission is strategy-proof to him and rank mechanisms as follows.

\begin{definition}\label{def-more-immune}
Mechanism $ \varphi$ is \textbf{more immune to strategic admissions} than mechanism $ \psi$ if 
\begin{itemize}
\item for each environment $ (\succ,q)$, if the admission to school $ s$ is strategy-proof to student $ i$ via $ \psi$, then the admission to school $ s$ is also strategy-proof to $i$ via $ \varphi$ and, \vspace{0.2cm}
\item there is an environment $(\succ,q)$ where the admission to some school $ s$ is strategy-proof to some student $ i$ via $ \varphi$, but not via $ \psi$.
\end{itemize}
\end{definition}

Observe that the immunity relation is transitive as it is based on set inclusion.

In the following section, we apply this definition to the mechanisms before and after the reforms. Subsequently, we discuss its explanatory power compared to the state-of-the-art notion.

\subsection{Reforms and immunity to strategic admissions}
\subsubsection*{In England and Wales}
According to the field observation (see Table \ref{reforms}) more than 50 areas in England and Wales have replaced different constrained versions of $FPF$ with constrained $ GS$. In the following theorem, we show that they have replaced mechanisms with more immune alternatives.

\begin{theorem} Suppose that there are at least $ k$ schools where $ k>1$ and at least one first-preference-first school. Then $GS^{k} $ is more immune to strategic admissions than $ FPF^{k}$. \label{pro2}
\end{theorem}

See the appendix for the proof. 
Kingston upon Thames has replaced a constrained $ FPF$ with a constrained $ GS$ but with a longer list. This replacement also resulted in a more immune mechanism:

\begin{corollary} Let $ k> \ell$ and suppose that there are at least $ \ell$ schools. Then $ GS^{k}$ is more immune to strategic admissions than $FPF^{\ell}$. \label{cor1}
  \end{corollary}
 
 See the appendix for the proof.
 
  \subsubsection*{In Chicago and Denver} In 2009, the Chicago Selective High Schools moved from a constrained Boston to a constrained serial dictatorship. A similar replacement has been observed in Denver and 4 other cities in England. As the Boston mechanism is a special case of $FPF$, and $SD$ is a special case of $GS$, we also explain this replacement.

  \begin{corollary} Let $ k\geq \ell$ and suppose that there are at least $ k$ schools. Then $ GS^{k}$ is more immune to strategic admissions than $ \beta^{\ell}$.
  \end{corollary}

  \subsubsection*{In Chicago and Ghana} In 2010, the Chicago Selective High Schools again replaced its constrained Serial Dictatorship with a version with longer list. In 2007, the Ghanaian Secondary Schools undertook a similar change, from a constrained $ GS$ to a version with longer list and extending the list again in 2008. These type of changes have also been observed in Newcastle and Surrey in England.

 \begin{theorem} Let $ k>\ell$ and suppose that there are at least $k$ schools. Then $ GS^{k}$ is more immune to strategic admissions than $ GS^{\ell}$. \label{pro3}
\end{theorem}

See the appendix for the proof. In the following table, we list all reforms in school choice systems. We also feature those that are comparable \'a la \cite{pathak2013} and those that are not (see section 4 for the results). 

\begin{table}[]
\caption{School Admissions Reforms (documented in \cite{pathak2013})}
\label{reforms}
\small
\begin{tabular}{@{}llllll@{}}
\toprule \toprule
Allocation system  & Year & From & To & \begin{tabular}[c]{@{}l@{}}Manipulable?\\ (More or less?)\end{tabular} & \begin{tabular}[c]{@{}l@{}}Immune?\\ (More or less?)\end{tabular} \\ \midrule
Boston Public School (K, 6, 9) & 2005 & Boston & GS & Less  & More \\
Chicago Selective High Schools & 2009 & Boston$^4$ & SD$^4$ & Less  & More \\
& 2010 & SD$^4$ & SD$^6$ & Less  & More \\
Ghana---Secondary schools& 2007 & GS$^3$ & GS$^4$ & Less  & More \\
& 2008 & GS$^4$ & GS$^6$ & Less  & More \\
Denver Public Schools  & 2012 & Boston$^2$ & GS$^5$ & Less  & More \\
Seattle Public Schools & 1999 & Boston & GS & Less  & More \\
& 2009 & GS   & Boston & More  & Less \\
England&&&  & &\\
\quad Bath and North East Somerset & 2007 & FPF$^3$& GS$^3$ & Not comparable  & More \\
\quad Bedford and Bedfordshire & 2007 & FPF$^3$& GS$^3$ & Not comparable  & More \\
\quad Blackburn with Darwen  & 2007 & FPF$^3$& GS$^3$ & Not comparable  & More \\
\quad Blackpool& 2007 & FPF$^3$& GS$^3$ & Not comparable  & More \\
\quad Bolton & 2007 & FPF$^3$& GS$^3$ & Not comparable  & More \\
\quad Bradford & 2007 & FPF$^3$& GS$^3$ & Not comparable  & More \\
\quad Brighton and Hove& 2007 & Boston$^3$ & GS$^3$ & Less  & More \\
\quad Calderdale   & 2006 & FPF$^3$& GS$^3$ & Not comparable  & More \\
\quad Cornwall & 2007 & FPF$^3$& GS$^3$ & Not comparable  & More \\
\quad Cumbria& 2007 & FPF$^3$& GS$^3$ & Not comparable  & More \\
\quad Darlington   & 2007 & FPF$^3$& GS$^3$ & Not comparable  & More \\
\quad Derby  & 2005 & FPF$^4$& GS$^4$ & Not comparable  & More \\
\quad Devon  & 2006 & FPF$^3$& GS$^3$ & Not comparable  & More \\
\quad Durham & 2007 & FPF$^3$& GS$^3$ & Not comparable  & More \\
\quad Ealing & 2006 & FPF$^6$& GS$^6$ & Not comparable  & More \\
\quad East Sussex  & 2007 & Boston$^3$ & GS$^3$ & Less  & More \\
\quad Gateshead& 2007 & FPF$^3$& GS$^3$ & Not comparable  & More \\
\quad Halton & 2007 & FPF$^3$& GS$^3$ & Not comparable  & More \\
\quad Hampshire& 2007 & FPF$^3$& GS$^3$ & Not comparable  & More \\
\quad Hartlepool   & 2007 & FPF$^3$& GS$^3$ & Not comparable  & More \\
\quad Isle of Wright   & 2007 & FPF$^3$& GS$^3$ & Not comparable  & More \\
\quad Kent   & 2007 & Boston$^3$ & GS$^4$ & Less  & More \\
\quad Kingston upon Thames   & 2007 & FPF$^3$& GS$^4$ & Not comparable  & More \\
\quad Knowsley & 2007 & FPF$^3$& GS$^3$ & Not comparable  & More \\
\quad Lancashire   & 2007 & FPF$^3$& GS$^3$ & Not comparable  & More \\
\quad Lincolnshire & 2007 & FPF$^3$& GS$^3$ & Not comparable  & More \\
\quad Luton  & 2007 & FPF$^3$& GS$^3$ & Not comparable  & More \\
\quad Manchester   & 2007 & FPF$^3$& GS$^3$ & Not comparable  & More \\
\quad Merton & 2006 & FPF$^6$& GS$^6$ & Not comparable  & More \\
\quad Newcastle& 2005 & Boston$^3$ & GS$^3$ & Less  & More \\
& 2010 & GS$^3$ & GS$^4$ & Less  & More \\
\quad North Lincolnshire & 2007 & FPF$^3$& GS$^3$ & Not comparable  & More \\
\quad North Somerset   & 2007 & FPF$^3$& GS$^3$ & Not comparable  & More \\
\quad North Tyneside   & 2007 & FPF$^3$& GS$^3$ & Not comparable  & More \\
\quad Oldham & 2007 & FPF$^3$& GS$^3$ & Not comparable  & More \\
\quad Peterborough & 2007 & FPF$^3$& GS$^3$ & Not comparable  & More \\
\quad Plymouth & 2007 & FPF$^3$& GS$^3$ & Not comparable  & More \\
\quad Poole  & 2007 & FPF$^3$& GS$^3$ & Not comparable  & More \\
\quad Portsmouth   & 2007 & FPF$^3$& GS$^3$ & Not comparable  & More \\
\quad Richmond & 2005 & FPF$^6$& GS$^6$ & Not comparable  & More \\
 \midrule
&&&  & & \textit{(Continued)}
\end{tabular}
\end{table}

\begin{table}[]
\caption*{\textbf{Table 1.} School Admissions Reforms (\emph{Continued})}
\small
\begin{tabular}{@{}llllll@{}}
\toprule \toprule
Allocation system  & Year & From & To & \begin{tabular}[c]{@{}l@{}}Manipulable?\\ (More or less?)\end{tabular} & \begin{tabular}[c]{@{}l@{}}Immune?\\ (More or less?)\end{tabular} \\ \midrule
 \\
\quad Sefton primary   & 2007 & Boston$^3$ & GS$^3$ & Less  & More \\
\quad Sefton secondary & 2007 & FPF$^3$& GS$^3$ & Not comparable  & More \\
\quad Slough & 2006 & FPF$^3$& GS$^3$ & Not comparable  & More \\
\quad Somerset & 2007 & FPF$^3$& GS$^3$ & Not comparable  & More \\
\quad South Gloucestershire  & 2007 & FPF$^3$& GS$^3$ & Not comparable  & More \\
\quad South Tyneside   & 2007 & FPF$^3$& GS$^3$ & Not comparable  & More \\
\quad Southhampton & 2007 & FPF$^3$& GS$^3$ & Not comparable  & More \\
\quad Stockton & 2007 & FPF$^3$& GS$^3$ & Not comparable  & More \\
\quad Stoke-on-Trent   & 2007 & FPF$^3$& GS$^3$ & Not comparable  & More \\
\quad Suffolk& 2007 & FPF$^3$& GS$^3$ & Not comparable  & More \\
\quad Sunderland   & 2007 & FPF$^3$& GS$^3$ & Not comparable  & More \\
\quad Surrey & 2007 & FPF$^3$& GS$^3$ & Not comparable  & More \\
& 2010 & GS$^3$ & GS$^6$ & Less  & More \\
\quad Sutton & 2006 & FPF$^6$& GS$^6$ & Not comparable  & More \\
\quad Swindon& 2007 & FPF$^3$& GS$^3$ & Not comparable  & More \\
\quad Tameside & 2007 & FPF$^3$& GS$^3$ & Not comparable  & More \\
\quad Telford and Wrekin & 2007 & FPF$^3$& GS$^3$ & Not comparable  & More \\
\quad Torbay & 2007 & FPF$^3$& GS$^3$ & Not comparable  & More \\
\quad Warrington   & 2007 & FPF$^3$& GS$^3$ & Not comparable  & More \\
\quad Warwickshire & 2007 & FPF$^7$& GS$^7$ & Not comparable  & More \\
\quad Wilgan & 2007 & FPF$^3$& GS$^3$ & Not comparable  & More \\
Wales &&&  & &\\
\quad Wrexham County Borough & 2011 & FPF$^3$& GS$^3$ & Not comparable  & More \\\midrule
\end{tabular}
\end{table}

\subsubsection*{In college admission in China}
Starting from 2001, most of the Chinese provinces changed their mechanisms from the Boston mechanism to various other parallel mechanisms. Consider two Chinese mechanisms: one with parameter $e$ and the other with $e>e'$. For these two mechanisms we obtain the following comparison.

\begin{theorem}\label{Chinese}
Let $ e > e'$ and suppose that there is at least $ e$ schools. Then $ Ch^{(e)}$ is more immune to strategic admissions than $ Ch^{(e')}$.
\end{theorem}

See the appendix for the proof.


\section{Comparison with Manipulability} 

In this section, we introduce another notion --- manipulability due to \cite{pathak2013}, and compare it with our notion of immunity to strategic admissions. 
\begin{definition}\citep{pathak2013}. Let $ \varphi$ and $ \psi$ be two mechanisms. Then mechanism $\varphi$ is \textbf{less manipulable} than mechanism $\psi$ if
\begin{itemize}
\item in each environment $ (\succ,q)$, each preference profile $ P$ that is vulnerable under $\varphi$, i.e., there exists some student $i$ and some misreport $P_i'$ such that, \begin{equation*}
 \varphi_{i}(P_i',P_{-i},\succ,q) \mathrel{P_i}\varphi_{i}(P,\succ,q), 
\end{equation*} is also vulnerable under $\psi$ and,
\item there is an environment where a preference profile is vulnerable under $ \psi$ but not under $ \varphi$.
\end{itemize}
\end{definition}
Broadly, whenever a student has a profitable manipulation at a preference profile $ P$ under $ \varphi$, then some --- possibly other --- student has a profitable manipulation under $ \psi$, while the reverse is not true in some environment. It is important to note that a profitable manipulation of one student is enough to declare a preference profile as vulnerable under a mechanism. Comparing mechanisms with respect to certain property profile by profile is common in the literature (a notable example is \cite{kesten2006two}), but it has two important limitations when applied to our case.


\subsection{Limitation in England and Wales}
First, we prove that the notion of manipulability does not explain the reforms in England and Wales. Recall that the education officials have replaced the constrained First-Preference-First with a constrained Gale-Shapley. Indeed, we provide a counterexample to Proposition 3 in \cite{pathak2013}. 
\vspace{0.4cm}
\begin{claim}[Proposition 3 --- \cite{pathak2013}] Suppose that there are at least $ k$ schools where $ k>1$. Then, $GS^{k}$ is less manipulable than $FPF^{k}$.
\end{claim}
\vspace{0.4cm}
We provide a counterexample to this claim. We specify the relevant part of the priorities such that the sign $ \vdots$ after student $ i$ indicates that the part below $ i$ is arbitrary and omitted.
\begin{example} 
Counterexample.
\end{example} Suppose that there are seven students and seven schools. Each school has one seat: $ q_{s}=1$ for each school $ s$. Let $ P$ and $ \succ$ be as specified below.
\begin{center}
\begin{tabular}{c c c c c c c|c c c c c c c |c}
$ P_{1}$ & $ P_{2}$ & $ P_{3}$ & $ P_{4}$ & $ P_{5}$ & $ P_{6}$ & $P_{7} $ & $ \succ_{s_{1}}$ & $ \succ_{s_{2}}$ & $ \succ_{s_{3}}$ & $ \succ_{s_{4}}$ & $ \succ_{s_{5}}$ & $ \succ_{s_{6}}$ & $ \succ_{s_{7}}$ & $ \succ'_{s_{5}}$\\\hline
$ s_{1}$ & $ s_{1}$ & $s_{2}$ & $ s_{3}$ & $s_{5}$ & $ s_{4}$ & $ s_{6}$ & $ 2$ & $ 3$ & $ 4$ & $ 7$ & $ 6$ & $ 6$ & $ 5$ & $ 5$\\
$ s_{2}$ & $ \emptyset $ & $\emptyset $ & $ \emptyset $ & $s_{7}$ & $ s_{5}$ & $ s_{4}$ & $ \vdots$ & $ \vdots$ & $ \vdots$ & $ 1$ & $ 5$ & $ 7$ & $ \vdots$ & $ 6$\\
$ s_{3}$ & $ $ & $ $ & $ $ & $\emptyset$ & $ s_{6}$ & $ \emptyset$ & $ $ & $ $ & $ $ & $ 6$ & $ \vdots$ & $ \vdots$ & & $ \vdots$ \\
$ s_{4}$ & $ $ & $ $ & $ $ & $ $ & $ \emptyset $ & $  $ & $ $ & $ $ & $ $ & $ \vdots$ & $ $ & $ $ & $ $ \\
$\emptyset$ & $ $ & $ $ & $ $ & $ $ & $ $ & $  $ & $ $ & $ $ & $ $ & $  $ & $  $ & $  $ & $ $ \\
\end{tabular}
\end{center} 

Suppose that school $ s_{5}$ is the only first-preference-first school. Then the matching is as follows:
\begin{equation*}
FPF^{3}(P,\succ,q)=GS^{3}(P,\succ,q)=\begin{pmatrix}
1 & 2 & 3 & 4 & 5 & 6 & 7 \\
\emptyset & s_{1} & s_{2} & s_{3} & s_{5} & s_{4} & s_{6}
\end{pmatrix}.
\end{equation*}

Every student except student $ 1$ received his most preferred school. Therefore only student $ 1$ could potentially manipulate each of $ FPF^{3}$ and $ GS^{3}$. Let $ P_{1}^{s_{4}}$ denote student $ 1$'s preference relation where school $ s_{4}$ is the only acceptable school. Then
\begin{equation*}
GS^{3}(P^{s_{4}}_{1},P_{-1},\succ,q)=\begin{pmatrix}
1 & 2 & 3 & 4 & 5 & 6 & 7 \\
s_{4} & s_{1} & s_{2} & s_{3} & s_{7} & s_{5} & s_{6}
\end{pmatrix}.
\end{equation*}
By reporting the preference relation $ P^{s_{4}}_{1}$ to $ GS^{3}$, student $ 1$ obtained an acceptable school $s_4$ but is unmatched when he reports his true preference relation $ P_{1}$. Therefore the profile $ P$ is vulnerable under $ GS^{3}$. However, because school $ s_{5}$ is a first-preference-first school and that student $ 5$ has ranked it higher than student $ 6$, we have (where $ \succ_{s_{5}}$ is replaced by $ \succ'_{s_{5}}$)
\begin{equation*}
FPF^{3}(P^{s_{4}}_{1},P_{-1},\succ,q)=\begin{pmatrix}
1 & 2 & 3 & 4 & 5 & 6 & 7 \\
\emptyset & s_{1} & s_{2} & s_{3} & s_{5} & s_{6} & s_{4}
\end{pmatrix}.
\end{equation*}
By reporting $ P^{s_{4}}_{1}$ to $ FPF^{3}$ student $ 1$ is unmatched, the same as when he reports his true preferences. It can be verified that it is enough to check for manipulations by ranking schools first. In addition, student $ 1$ cannot misrepresent his preferences to obtain a seat at school $ s_{1}$, $ s_{2}$ and $ s_{3}$. Therefore the profile $ P$ is not vulnerable under $ FPF^{3}$. 

The intuition is that when student $ 1$ claims school $ s_{4}$ he causes the rejection of student $ 6$. Then student $ 6$ claims school $ s_{5}$. Under $ GS^{3}$ student $ 5$ is rejected from school $ s_{5}$ and he applies to school $ s_{7}$, ending the process. However under $ FPF^{3}$, school $ s_{5}$ is a first-preference-first school which student $ 5$ has ranked first. This time it is student $ 6$ who is rejected from school $ s_{5}$. Then he claims school $ s_{6}$ and causes the rejection of student $ 7$. Ultimately, student $ 7$ claims school $ s_{4}$ and takes it back from student $ 1$. End of the example.
\smallskip

Nevertheless, when each school is a first-preference-first school then the comparison --- between the constrained Boston and the constrained Gale-Shapley --- is valid. 

\begin{proposition}\citep{pathak2011}
Suppose that there are at least $ k$ schools where $ k>1$. Then $GS^{k} $ is less manipulable than $ \beta^{k}$.
\end{proposition}

\subsection{Limitation in Chicago} 

The other limitation of manipulability is that it is insensitive when applied to constrained versions of strategy-proof mechanisms: it only bites for those profiles where the constraint becomes not binding. Indeed, if at a particular profile a constrained mechanism is not manipulable, then at this profile it must be very close to its unconstrained version. We formalize this intuition for the serial dictatorship mechanism used in Chicago.

\begin{proposition}\label{char_man_SDk}
A preference profile $ P$ is not vulnerable under $SD^k$ if, and only if $SD^k(P)=SD(P)$.
\end{proposition}

\begin{proof}
The if part is straightforward: if $SD^k(P)=SD(P)$, i.e. if the constraint is not binding, then at $SD^k(P)$ each student receives the best available school among remaining ones and cannot profitably misreport his preferences.

We prove the only if part by contraposition. Suppose that $ SD^{k}(P)\neq SD(P)$ and consider the highest priority student $i$ for whom $SD^k_i(P)\neq s=SD_i(P)$. Each student with higher priority than $i$ received under $SD^k(P)$ the same school as under $SD(P)$, therefore under $SD^k(P)$ student $i$ had the same choice set of remaining schools as under $SD(P)$. The only way $i$ missed school $s$ under $SD^k(P)$ is if the constraint $k$ was binding for him: each of his top $k$ schools were already assigned, and school $s$ was not listed. However, school $s$ still had available seats, and $i$ could profitably manipulate $SD^k$ at $P$ by listing school $s$ as one of his top $k$ schools.
\end{proof}

For $k=1$ this result also applies to Boston (constrained and unconstrained) with a common priority: at a given preference profile $P$, Boston is not manipulable if and only if its outcome coincides with the outcome of the unconstrained serial dictatorship at $P$.

For realistic profiles, however, the constraint is almost guaranteed to be binding at least for one student, and thus the constrained mechanism remains manipulable at this profile. This occurs, for instance, when the preferences of students are correlated, as often is the case. Next we generalize the Chicago example presented in the introduction. We show that when students have tier preferences, constrained versions of serial dictatorship are always manipulable. In general, the more tiers there are, the more binding the constraint will be.

\begin{example} Serial dictatorship and tier preferences.
\end{example}
Consider a school choice problem with $n$ students and $m$ schools, each two schools $s,s'\in S$ having the same capacity $q_s=q_{s'}$ and a common priority ranking $\succ_s=\succ_{s'}$. 

Assume that students have tier preferences: the set of schools $S$ is partitioned into $t>1$ sets $S_1,S_2,...,S_t$ and each student $i\in I$ prefers each school in $ S_{j}$ from a higher tier $j<t$ over each school in $S_{j+1}$ from a lower tier.\footnote{\cite{coles2013} observed that the academic job market has this structure and referred to it as block correlated preferences.} Assume also that each student finds each school acceptable and that there is shortage of seats, $n>q\times m$.

It is straightforward to show that whenever the number of schools in the upper tiers is at least as large as the constraint, $|S|-|S_t|\geq k$, the constrained serial dictatorship $SD^k$ is always manipulable. Otherwise, if every student reports truthfully, then some students are unmatched while acceptable schools in the last tier $S_t$ are unassigned. By Proposition \ref{char_man_SDk} this is necessary and sufficient for manipulability of $SD^k$. This is why for correlated preferences the manipulability criterion is not sensitive and the constrained serial dictatorship mechanism is as manipulable as the Boston mechanism.

In contrast, our notion remains sensitive in this domain: as the constraint changes, the immunity to strategic admissions changes as well. We formulate this as a proposition.

\begin{proposition}
Let $ k\geq 1$ and $ (\succ,q)$ an environment where schools have a common priority. Let the capacities of the schools be increasingly ordered $ q_{1}\leq q_{2} \leq\hdots \leq  q_{\lvert S\lvert}$ and $ \hat{q}=q_{1}+\hdots + q_{k}$. Then, the mechanism $ SD^{k}$ is strategy-proof for the $\hat{q}$-highest priority students. \label{lemma1}
\end{proposition}

See the appendix for the proof. Proposition \ref{lemma1} is formulated for the entire domain of preferences, and it remains true in the domain of tier preferences. Therefore, in the example, the share of students for whom the admission to each school is strategy-proof is $q\times k/n$.

The schools in the upper tiers are more protected from strategic admissions. By switching from $ \beta^{4}$ to $ SD^{4}$ in 2009, and to $ SD^{6}$ in 2010, the strategy-proof admissions to Chicago elite schools, increased from $ 4\%$ to $ 24\%$ and eventually to $ 100\%$, respectively.

\section{Strategic admissions in equilibrium} 
In this section, we develop a more refined concept of immunity to strategic admissions. Previously, we called admission of student $i$ to school $s$ strategic if there exists a profile and a profitable deviation for $i$ that places $i$ to $s$. But this deviation did not need to be optimal, same as the reports of other students did not need to be optimal, and thus not rationalizable. Now we require the strategies to be mutually optimal.

Let us motivate this with an example.

\begin{example} Equilibrium in Boston. \end{example}

Suppose there are three students $ i$, $ j$ and $ k$ and three schools $ s_{1}$, $ s_{2}$ and $ s_{3}$ with one seat each. The preferences and the priorities are as follows.
\begin{center}
\begin{tabular}{c c c|c c c}
   $ P_{i}$& $ P_{j}$ & $ P_{k}$ & $ \succ_{s_{1}}$ & $ \succ_{s_{2}}$ & $ \succ_{s_{3}}$ \\\hline
   $ s_{1}$ & $ s_{1}$ & $ s_{1}$ & $ j$ & $ j$ & $ j$\\
   $s_{2}$ & $ s_{2}$ & $ s_{2}$ & $k $ & $ k$ & $k $\\
   $s_{3}$ & $ s_{3}$ & $ s_{3}$ & $i $ & $ i$ & $i $\\
   $\emptyset$ & $\emptyset$ & $ \emptyset$ & $ $ & $ $ & $ $\\

\end{tabular}
\end{center}
Let us consider the Boston mechanism. Its outcome for this problem is specified as follows. 
\begin{equation*}
\beta (P,\succ,q)= \begin{pmatrix}
 i &  j &  k \\  s_{3} &  s_{1} &  s_{2}
\end{pmatrix}.
\end{equation*}
Suppose instead that student $ i$ reports the preference relation $ P_{i}^{s_{2}}$ where he ranks school $ s_{2}$ first. If student $ j$ and $ k$ report truthfully as in $ P$, we have
\begin{equation*}
\beta (P^{s_{2}}_{i},P_{-i},\succ,q)= \begin{pmatrix}
 i &  j &  k \\  s_{2} &  s_{1} &  s_{3}
\end{pmatrix}.
\end{equation*} According to the notion developed earlier, the admission to school $ s_{2}$ is not strategy-proof to student $ i$ via the Boston mechanism. However, it is not a best response for student $ k$ to report truthfully $ P_{k}$, when student $ i$ reports $ P^{s_{2}}_{i}$. Student $ i$ has the lowest priority at every school. The lack of strategy-proof admissions for student $ i$ stems from the fact that other students report their preferences truthfully without best-responding. End of example.
\smallskip

This example demonstrates that sometimes the admission to some schools is not strategy-proof to some students only because others are not best responding. This type of admissions may disappear when we require best responses.  To take these best responses into account we introduce an equilibrium concept. If we fix an environment $ (\succ,q)$, any mechanism $ \varphi$ induces a normal form game such that the students are the players, the strategies are the preferences and the outcome function is $\varphi(.,\succ,q)$. Then, a strategy profile $ P'$ is a \textbf{Nash equilibrium} of the game $(I,P,\varphi(.,\succ,q))$ if for each student $ i$, $P'_{i}$ is a best response to $ P'_{-i}$.\footnote{That is, for each student $ i$, there is no strategy $ P''_{i}$ such that $\varphi_{i}(P''_{i},P'_{-i},\succ,q)\mathrel{P_{i}}\varphi_{i}(P',\succ,q).$} When there is no ambiguity on the environment, we denote the game as $ (P,\varphi)$.

\begin{definition}\label{def-str-acc-eq} Let $ \varphi$ be a mechanism and $ (\succ,q)$ an environment. The \textbf{admission to school $ s$ is strategy-proof to student $ i$ via mechanism $ \varphi$ \underline{in equilibrium}} if there is no preference profile $ P$ and a preference relation $ P'_{i}$ such that 
\begin{enumerate}
\item[(1)] $ (P'_{i},P_{-i})$ is a \underline{Nash equilibrium} of the game $ (P,\varphi)$ and
\item[(2)] $ s=\varphi_{i}(P'_{i},P_{-i},\succ,q) \mathrel{P_{i}}\varphi_{i}(P,\succ,q)$
\end{enumerate}
\end{definition}

Note that when the admission to school $ s$ is strategy-proof to student $ i$ via $ \varphi$, there is no preference profile and a deviation that satisfy the condition (2). Therefore the admission to school $ s$ is strategy-proof to student $ i$ via $ \varphi$ in equilibrium. Clearly, Definition \ref{def-str-acc-eq} is more stringent than Definition \ref{def-str-pr-adm}. Let us now use this notion to define a ranking criterion analogous to the one defined in Definition \ref{def-more-immune}.

\begin{definition}
Mechanism $ \varphi$ is \textbf{\underline{strongly} more immune to strategic admissions} than mechanism $ \psi$ if 
\begin{itemize}
\item for each environment $ (\succ,q)$, if the admission to school $ s$ is strategy-proof to student $ i$ via $ \psi$ \underline{in equilibrium}, then the admission to school $ s$ is also strategy-proof to him via $ \varphi$ \underline{in equilibrium} and, \vspace{0.2cm}
\item there is an environment $ (\succ,q)$ where the admission to some school $ s$ is strategy-proof to some student $ i$ via $ \varphi$ \underline{in equilibrium}, but not via $ \psi$.
\end{itemize}
\end{definition}
\vspace{0.5cm}

Despite this stringent notion, the main results for constrained Gale-Shapley mechanism and the First-Preference-First mechanism remain true.

\begin{theorem}\label{Th-GS-FPF-eq}
Let $ k>\ell$ and suppose that there are at least $ k$ schools and at least one first-preference-first school. Then
\begin{itemize}
    \item $ GS^{k}$ is strongly more immune to strategic admissions than $ GS^{\ell}$,
    \item $ GS^{k}$ is strongly more immune to strategic admissions than $ FPF^{k}$.
\end{itemize} \label{thm4}
\end{theorem}
See the appendix for the proof. In contrast, the prior ranking of the Chinese mechanisms does not hold anymore. 

\begin{proposition}\label{Prop-Ch-eq}
There is $ e > e'$ and at least $ e$ schools such that the the mechanism $ Ch^{(e)}$ is not strongly less immune to strategic admissions than $ Ch^{(e')}$.
\end{proposition} 
\begin{proof}
We prove by the following example and the mechanisms $ Ch^{2}$ and $ Ch^{1}=\beta$.

Suppose that there are 4 students and 4 schools. Let $ (P^{\ast},\succ,q)$ be a problem where each school has one seat and $P $ and $ \succ$ are specified as follows.
\begin{center}
\begin{tabular}{c c c c c c|c c c c}
   $ P^{\ast}_{i}$& $ P^{\ast}_{j}$ & $ P^{\ast}_{k}$ & $ P^{\ast}_{m}$ & $ P^{\ast}_{t}$ & $ P'_{i}$ & $ \succ_{s_{1}}$ & $ \succ_{s_{2}}$ & $ \succ_{s_{3}}$ & $ \succ_{s_{4}}$\\\hline
   $ s_{3}$ & $ s_{3}$ & $ s_{2}$ & $ s_{2}$ & $ s_{4}$ & $ s_{1}$ & $ k$ & $ i$ & $ j$ & $ t$ \\
   $ s_{4}$ & $ s_{2}$ & $ s_{1}$ & $ s_{3}$ & $\emptyset$ & $ s_{2}$ & $ m$ & $ k$ & $ \vdots$ & $ \vdots$ \\
   $ s_{1}$ & $\emptyset$ & $ \emptyset$ & $ s_{1}$ & $ $ & $ \emptyset$ & $ j$ & $ \vdots$ & $ $ & $ $ \\
   $ s_{2}$ &  &  & $ \emptyset$ & &  & $ t$ & $ $ & $ $ & $ $ \\
$ \emptyset $ &  &  & $  $ &  & & $ i$ & $ $ & $ $ & $ $ \\

\end{tabular}
\end{center}
Then we have
\begin{equation*}
Ch^{(2)}(P^{\ast},\succ,q)=\begin{pmatrix} i& j& k& m &t \\ \emptyset & s_{3} & s_{2} & s_{1} & s_{4}
\end{pmatrix}.
\end{equation*}
Suppose that student $ i$ reports the preference relation $ P'_{i}$. We show that $ (P'_{i},P^{\ast}_{-i})$ is a Nash equilibrium of the game $ (P^{\ast},Ch^{(2)})$. First, 
\begin{equation}
Ch^{(2)}(P'_{i},P^{\ast}_{-i},\succ,q)=\begin{pmatrix} i& j& k& m &t \\ s_{1} & s_{3} & s_{2} & \emptyset & s_{4}
\end{pmatrix}. \label{equation1}
\end{equation}
In the matching in equation \ref{equation1}, every student, other than $ i$ and $ m$, is matched to hisr most preferred school under $ P^{\ast}$. Thus we need to check that it is a best response for student $ i$ and $ m$. School $ s_{3}$ and $ s_{4}$ are assigned to the highest priority students. Therefore, student $ i$ cannot get a seat at each of these schools by reporting a preference relation other than $ P'_{i}$. Let us consider now student $ m$. In any strategy where he did not include school $ s_{1}$ among the top two acceptable schools, the outcome is the matching in equation \ref{equation1}. Suppose that he uses a strategy $ P'_{m}$ where he includes school $ s_{1}$ among the top two acceptable schools. Then, 
\begin{equation*}
Ch^{(2)}(P'_{i},P'_{m},P^{\ast}_{-\{i,m\}},\succ,q)=\begin{pmatrix} i& j& k& m &t \\ s_{2} & s_{3} & s_{1} & \emptyset & s_{4}
\end{pmatrix},
\end{equation*}
where student $ m$ remains unmatched. Therefore, student $ i$ and $ m$ do not have a profitable deviation and $ (P'_{i},P^{\ast}_{-i})$ is a Nash equilibrium of the game $ (P^{\ast},Ch^{(2)})$. Since
\begin{equation*}
s_{1}=Ch_{i}^{(2)}(P'_{i},P^{\ast}_{-i},\succ,q) \mathrel{P^{\ast}_{i}} Ch^{(2)}_{i}(P^{\ast},\succ,q),
\end{equation*}
then the admission to school $ s_{1}$ is not strategy-proof to student $ i$ via $ Ch^{(2)}$ in equilibrium. 

Next, we show that the admission to school $ s_{1}$ is strategy-proof to student $ i$ via $ Ch^{(1)}=\beta$ in equilibrium. Suppose that for some preference profile $ P$ and $ P'_{i}$, we have 
\begin{equation}
s_{1}=Ch_{i}^{(1)}(P'_{i},P_{-i},\succ,q) \mathrel{P_{i}} Ch^{(1)}_{i}(P,\succ,q).\label{equation2}
\end{equation}
We now show that $(P'_{i},P_{-i}) $ is not a Nash equilibrium of the game $ (P,\beta)$. This will complete the proof showing that student $ i$ does not have a strategic admission to school $ s_{1}$ via $ \beta$. 
The Boston mechanism produces a Pareto optimal matching with respect to reported preferences. Therefore equation \ref{equation2} implies that some student $ j'$ is worse-off under $ \beta(P'_{i},P_{-i},\succ,q)$ compared to $ \beta(P,\succ,q)$. We consider two cases depending on what $j'$ gets:

Case 1: Student $ j'$ is matched to his first choice school under $ \beta(P,\succ,q) $, denoted by $ s$. Then student $ j'$ is the highest priority student among those who ranked school $ s$ first. Since $ j'$ is worse-off, and thus not matched to school $ s$ under $ \beta(P'_{i},P_{-i},\succ,q)$, then student $ i$ has ranked school $ s$ first under $ P'_{i}$ and is matched to it. By equation \ref{equation2}, $ s=s_{1}$ which contradicts the fact that student $ j'$ has higher priority than student $ i$ under $ \succ_{s_{1}}$.

Case 2: Student $ j'$ is not matched to his first choice school under $ \beta(P,\succ,q) $. Let $ s=\beta_{j'}(P,\succ,q)$. Then no student ranked school $ s$ first under $ P$. Let $ P^{s}_{j'}$ be a preference relation where he has ranked school $ s$ first. We claim that $ \beta_{j'}(P'_{i},P^{s}_{j'},P_{-\{i,j'\}},\succ,q)=s$. Suppose, to the contrary, that this is not the case. Then student $ 1$ has ranked school $ s$ first under $ P'_{i}$, and is the only student who has ranked it first under $ (P'_{i},P_{-i})$. Thus $ s=\beta_{i}(P'_{i},P_{-i},\succ,q)=s_{1}$. Since student $ i$ has lower priority than student $ j'$ under $ \succ_{s_{1}}$, $\beta_{j'}(P'_{i},P^{s}_{j'},P_{-\{i,j'\}},\succ,q)=s $, contradicting our assumption that student $ j'$ is not matched to school $ s$. 

Therefore student $ j'$ has a profitable deviation from $(P'_{i},P_{-i}) $, and $ (P'_{i},P_{-i})$ is not a Nash equilibrium of the game $ (P,\beta)$. Under $ (\succ,q)$, the admission to school $ s$ is strategy-proof to student $ i$ via $ Ch^{(1)}$ in equilibrium.
\end{proof}

\vspace{0.7cm}
\section{Conclusion} \label{conclusions}
\vspace{0.7cm}

All strategy-proof mechanisms are alike, each vulnerable mechanism is vulnerable in its own way. Compared to another, a mechanism can be vulnerable at a larger set of profiles as in \cite{pathak2013}, larger set of preference relations as in \cite{decerf2018manipulability}. A mechanism can also be manipulated by a larger set of agents, giving them stronger incentives to manipulate and causing worse consequences for others --- as measured by the range of outcomes that these manipulations induce. This is the focus of our paper.

We compare vulnerable mechanisms by what we call immunity: that is how many agents can manipulate them and to what extent. This notion is first introduced by \cite{som} to study the relation between favoring higher ranks and incentives in school choice. Here we argued that this metric represents the concerns of the policy-makers and the public that accompanied recent reforms in admissions systems around the world; and showed that each of these reforms made the mechanisms more immune. Namely, after each reform, in each environment (set of students, schools, priorities and capacities), for each student there are fewer schools that he can access by profitable manipulation, and each school is more protected from these manipulations.

Immunity also comes in an equilibrium version: when student $i$ manipulates the mechanism to get to a school at preference profile $P$, we require $P$ and $i$'s deviation to form an equilibrium in the game induced by the mechanism. This concept is arguably less realistic for markets where best responses is hard to expect, e.g., when the market is large, but it is a standard refining criterion for smaller problems. Our main results carry over to this equilibrium version of immunity for each mechanism except for the Chinese (Theorem \ref{Th-GS-FPF-eq}, Proposition \ref{Prop-Ch-eq}).

\cite{pathak2013}'s seminal paper was one of the first attempts to compare vulnerable mechanisms and it has generated a literature on other applications and extensions. \cite{ar2016} ranked voting rules by inclusion of the vulnerable preference relations of each agent $i$, that is the relations for which there exist preferences of others such that $i$ can manipulate. This notion was recently used by \cite{decerf2018manipulability} to rank constrained Boston and GS mechanisms. \cite{andersson2014} ranked budget balanced and fair rules by counting, for each preference profile, the number of agents who have profitable manipulations. They find rules that minimize (by inclusion) the number of agents and coalitions that can manipulate. In the same problem \cite{andersson2014least} find the rule that minimizes the maximal gain that an agent can get by manipulation. Next, \cite{chenandkesten2017a} and \cite{dur2018} used manipulability to compare mechanisms used in China and Taiwan, respectively. 

Another related notion is manipulability ranking criterion due to \cite{chen2016}: compared to mechanism A, they define mechanism B to be more manipulable if for each agent, at each profile where he can manipulate and get a particular outcome under mechanism A, he can do the same under mechanism B. This notion is useful for ranking all stable mechanisms in an intuitive sense. However, unlike the results in \cite{pathak2013} and in the current paper, these results cannot be attained using the equilibrium version of their manipulability criterion: in equilibrium each stable mechanism is as manipulable as another. As this notion also relies on the preference by preference and agent by agent comparison, it does not explain any of the reforms studied here (except the Chicago SHS). 

We should note that immunity to manipulation, regardless of the precise metric used to measure it, is not the final criterion in selecting a mechanism. Perhaps, the ultimate concern of the policy-makers and the parents is not the vulnerability itself, but rather the complexity of finding an optimal strategy. This complexity results in drawbacks such as higher number of mismatches, wastes, justified envy and overall dissatisfaction with the system. Surprisingly, the mechanism designers around the world are ready to tolerate certain levels of these drawbacks, but why that ever do it when a strategy-proof mechanism is readily available --- for example, the unconstrained version of the student-proposing deferred acceptance mechanism is strategy-proof \citep{roth1982a,dubins} --- remains obscure. 

\bibliographystyle{apalike}
\bibliography{biblio}
\section*{Appendix: Proofs}

\textit{\underline{Theorem \ref{pro2}}: Suppose that there are at least $ k$ schools where $ k>1$ and at least one first-preference-first school. Then $GS^{k} $ is more immune to strategic admissions than $ FPF^{k}$.}

\begin{proof}[Proof of Theorem \ref{pro2}] We prove the theorem by contraposition. Suppose that there are at least $ k$ schools and at least one first-preference-first school (in each environment).
Let an environment $ (\succ,q)$ be given and suppose that the admission to school $ s$ is not strategy-proof to student $ i$ via $ GS^{k}$. There is a preference profile $ P$ and a preference relation $ P'_{i}$ such that
\begin{equation}
s=GS^{k}_{i}(P'_{i},P_{-i},\succ,q) \mathrel{P_{i}} GS^{k}_{i}(P,\succ,q).\label{eq1}
\end{equation} 

We also show that the admission to school $s$ is not strategy-proof to $i$ via $FPF^{k}$. We are going to prove two facts and draw a lemma. 

\textbf{Fact 1:} $GS^{k}_{i}(P,\succ,q)=\emptyset$

 If $ GS^{k}_{i}(P,\succ,q)=s'\in S$ then school $ s'$ is one of the schools that student $ i$ has ranked among his top $ k$ schools under $ P_{i}$. Since $s \mathrel{P_{i}} GS^{k}_{i}(P,\succ,q)$, we have $ s \mathrel{P_{i}} s'$. Thus school $ s$ is one of the schools that student $ i$ has ranked among his top $k $ schools under $ P_{i}$. Then school $ s$ and school $ s'$ are acceptable under $ P^{k}_{i}$. By definition $ GS^{k}(P'_{i},P_{-i},\succ,q)=GS(P'^{k}_{i},P^{k}_{-i},\succ,q)$ and $ GS^{k}(P,\succ,q)=GS(P^{k},\succ,q)$. Therefore
\begin{equation}
GS_{i}(P'^{k}_{i},P^{k}_{-i},\succ,q)=s \mathrel{P^{k}_{i}} s'= GS_{i}(P^{k},\succ,q).\label{eq2}
\end{equation}
This means, contrary to the fact that $ GS$ is strategy-proof \citep{roth1982a,dubins}, that student $ i$ manipulates $GS $ at $ P^{k}$. Therefore, $ GS^{k}_{i}(P,\succ,q)=\emptyset$.

\textbf{Fact 2:} Student $ i$ did not rank school $ s$ among the top $ k$ schools under $ P_{i}$.

 Otherwise, school $ s$ is acceptable under $ P^{k}_{i}$ and the fact that $ GS_{i}(P^{k},\succ,q)=\emptyset$, by \textbf{Fact 1}, we would have
\begin{equation*}
GS_{i}(P'^{k}_{i},P^{k}_{-i},\succ,q)=s \mathrel{P^{k}_{i}} \emptyset=GS_{i}(P^{k},\succ,q),
\end{equation*}
also contradicting the fact that GS is strategy-proof. Because $ GS^{k}$ is individually rational, equation \ref{eq1} implies that school $ s$ is acceptable to student $ i$ under $ P_{i}$. Now because school $ s$ is acceptable under $ P_{i}$ but not under $ P^{k}_{i}$, student $ i$ has ranked more than $ k$ schools acceptable under $ P$. Let us state these results in a lemma that we use later.
\begin{lemma} Let us suppose that for some preference profile $ P$, a student $ i$ and $ P'_{i}$, we have
\begin{equation*}
s= GS^{k}_{i}(P'_{i},P_{-i},\succ,q) \mathrel{P_{i}} GS^{k}_{i}(P,\succ,q).
\end{equation*}
Then (i) student $ i$ has ranked more than $ k$ schools acceptable under $ P_{i}$, and (ii) school $ s$ is acceptable under $ P_{i}$ but not ranked among the top $ k$ schools under $ P_{i}$. \label{lemma}
\end{lemma}
We now prove that the admission to school $ s$ is not strategy-proof to student $ i$ via $ FPF^{k}$. For each $ \ell\leq k$, let $ s_{\ell}$ denote the school that student $ i$ has ranked at position $ \ell$ under $ P_{i}$. Let $ \mu=GS^{k}(P,\succ,q)$. Because $ \mu$ is stable under $ (P^{k},\succ,q)$ and that for each $ \ell\leq k$, $ s_{\ell} \mathrel{P^{k}_{i}} \mu(i)$ due to \textbf{Fact 1} (that is, $\mu(i)=\emptyset$) and Lemma \ref{lemma}, for each $ \ell\leq k$, we have $ \lvert\mu^{-1}(s_{\ell})\lvert =q_{s_{\ell}}$ and for each student $ j\in \mu^{-1}(s_{\ell})$, $ j\mathrel{\succ_{s_{\ell}}} i$. Construct a preference profile $ P^{\ast}$ as follows:

\begin{equation}
\begin{tabular}{c c}
$ P^{\ast}_{i}$ & $ P^{\ast}_{j\neq i}$ \\\hline
$ s_{1}$ & $ \mu(j)$\\
$ s_{2}$ & $ \emptyset$\\
$\vdots$ & \\
$ s_{k}$ &\\
$ s$ & \\
$ \emptyset$ &
\end{tabular} \label{eqFP}
\end{equation}

Each of the students in $ \mu^{-1}(s_{1})$ have higher priority than student $ i$ under $ \succ_{s_{1}}$ and have ranked it first under $ P^{\ast}$. Thus $ FPF^{k}_{i}(P^{\ast},\succ,q)\neq s_{1}$. More generally, for each $ \ell\leq k$, each student in $ \mu^{-1}(s_{\ell})$ has higher priority than student $ i$ under $ \succ_{s_{\ell}}$ and has ranked school $ s_{\ell}$ higher than student $ i$. Therefore $ FPF^{k}(P^{\ast},\succ,q)=\mu$ where student $ i$ is unmatched. Let $ P^{s}_{i}$ be a preference relation where student $ i$ finds only school $ s$ acceptable. We claim that $ FPF^{k}_{i}(P^{s}_{i},P^{\ast}_{-i},\succ)=s$. We consider two cases:
\par Case 1: $ \lvert \mu^{-1}(s) \lvert < q_{s}$. In that case, clearly $ FPF^{k}_{i}(P^{s}_{i},P^{\ast}_{-i},\succ,q)=s$ because no more than $ q_{s}$ students finds school $ s$ acceptable under $ (P^{s}_{i},P^{\ast k}_{-i})$.
\par Case 2: $ \lvert \mu^{-1}(s)\lvert = q_{s}$. In this case, we claim that there is at least one student in $ \mu^{-1}(s)$ who has lower priority than student $ i$ under $ \succ_{s}$. Suppose, to the contrary, that each student in $ \mu^{-1}(s)$ has higher priority than student $ i$ under $ \succ_{s}$.

By definition, $ GS^{k}(P,\succ,q)=GS(P^{k},\succ,q)$. Then $ \mu$ is stable under $ (P^{k},\succ,q)$. It is also stable under $ (P^{s}_{i},P_{-i}^{k},\succ,q)$. This is because, any student other than $ i$ does not have a justified envy and student $ i$'s envy is not justified --- every student in $ \mu^{-1}(s)$ has higher priority than $ i$ under $ \succ_{s}$.

By \cite{roth1986} the set of students who are matched is the same at all stable matchings. Since student $ i$ is unmatched under $ \mu$, we have $ GS_{i}(P^{s}_{i},P^{k}_{-i},\succ,q)=\emptyset$. In addition, using a result by \cite{roth1982a}, $ GS_{i}(P'^{k}_{i},P^{k}_{-i},\succ,q)=s$ implies that $ GS_{i}(P^{s}_{i},P^{k}_{-i},\succ,q)=s$. This contradicts the conclusion that $ GS_{i}(P^{s}_{i},P^{k}_{-i},\succ,q)=\emptyset$. Therefore there is at least one student in $ \mu^{-1}(s)$ who has lower priority than student $ i$ under $ \succ_{s}$.

Thus $ FPF^{k}_{i}(P^{s}_{i},P^{\ast}_{-i},\succ,q)=s$. Finally, because $ FPF^{k}_{i}(P^{\ast},\succ,q)=\emptyset$ and school $ s$ is acceptable under $ P^{\ast}_{i}$ by construction,
\begin{equation*}
s=FPF^{k}_{i}(P^{s}_{i},P^{\ast}_{-i},\succ,q) \mathrel{P^{\ast}_{i}} FPF_{i}^{k}(P^{\ast},\succ,q),
\end{equation*}
\noindent proving that the admission to school $s$ is not strategy-proof to student $ i$ via $ FPF^{k}$. Finally, we construct an environment where the admission to some school is strategy-proof to some student via $ GS^{k}$ but not via $ FPF^{k}$.

Let $ (\succ,q)$ be an environment such that schools have the same priority order and have one seat each. Since there is at least one first-preference-first school, let school $ s$ be one such a school. Let student $ i$ be the student who is ranked first, $ j$ second and $ m$ third, in the common priroity. Since $ k\geq 2$, the admission to school $ s$ is strategy-proof to student $ j$ via $ GS^{k}$. Let us now show that the admission to school $ s$ is not strategy-proof to student $ j$ via $ FPF^{k}$. This follows from Lemma \ref{lemma} and because he is always matched to one of his top two schools when he ranks at least two schools acceptable. Let $ P$ be a preference profile such that the components for $ i$, $ j$ and $ m $ are specified as below.   

\begin{center}
\begin{tabular}{c c c}
$ P_{i}$ & $ P_{j}$ & $ P_{m}$\\\hline
$ s'$ & $ s'$ & $ s$\\
$ s$ & $ s$ & $ s'$\\
$ \emptyset$ & $ \emptyset$ & $ \emptyset$
\end{tabular}
\end{center}
 Then $ FPF^{k}_{j}(P,\succ,q)=\emptyset$ because school $ s$ is a first-preference-first school for which student $ j$ did not rank as high as student $m$. Let $ P^{s}_{i}$ be a preference relation where $ s$ is the only acceptable school for student $ j$. Then $ FPF^{k}_{j}(P^{s}_{j},P_{-j},\succ,q)=s$. Therefore the admission to school $ s$ is not strategy-proof to student $ j$ via $ FPF^{k}$. 
\end{proof}

\textit{\underline{Theorem \ref{pro3}}: Let $ k>\ell$ and suppose that there are at least $k$ schools. Then $ GS^{k}$ is more immune to strategic admissions than $ GS^{\ell}$.}

\begin{proof}[Proof of Theorem \ref{pro3}] We prove the theorem by contraposition. Let $ k> \ell$ and suppose that there is at least $k$ schools. We show that $ GS^{k}$ is more immune to strategic admissions than $ GS^{\ell}$.
Fix an environment $ (\succ,q)$ and suppose that the admission to school $ s$ is not strategy-proof to student $ i$ via $ GS^{k}$. Then there is a preference profile $ P$ and a preference relation $ P'_{i}$ such that
\begin{equation*}
s=GS^{k}_{i}(P'_{i},P_{-i},\succ,q) \mathrel{P_{i}} GS^{k}_{i}(P,\succ,q).
\end{equation*}
 We show that the admission to school $ s$ is not strategy-proof to student $i$ via $GS^{\ell}$.

By Lemma \ref{lemma}, student $ i$ has ranked more than $ k$ schools acceptable under $ P_{i}$ and school $ s$ is acceptable under $ P_{i}$ but not ranked among the top $ k$ schools. For each $ \ell'\leq k$, let $s_{\ell'} $ denote the school that student $ i$ has ranked at position $ \ell'$ under $ P_{i}$.

By definition $ GS^{k}(P,\succ,q)=GS(P^{k},\succ,q)$. Let $ \mu=GS(P^{k},\succ,q)$. Because $ \mu$ is stable under $ (P^{k},\succ,q)$ and student $ i$ is unmatched under $ \mu$ by \textbf{Fact 1}, for each $ \ell'\leq k$, $ s_{\ell'} \mathrel{P_{i}^{k}} \mu(i)$ implies that for each student $ j\in \mu^{-1}(s_{\ell'})$, $ j\mathrel{\succ_{s_{\ell'}}} i$. Let $ P^{\ast}$ be a preference profile defined as follows:

\begin{equation}
\begin{tabular}{c c}
$ P^{\ast}_{i}$ & $ P^{\ast}_{j\neq i}$ \\\hline
$ s_{1}$ & $ \mu(j)$\\
$ s_{2}$ & $ \emptyset$\\
$\vdots$ & \\
$ s_{\ell}$ &\\
$ s$ &\\
$ \emptyset$ &
\end{tabular}\label{eqGS}
\end{equation}

Then $ GS^{\ell}(P^{\ast},\succ,q)=\mu$, where student $ i$ is unmatched. Let $ P^{s}_{i}$ be a preference relation where $ s$ is the only acceptable school for student $ i$. If $ \lvert \mu^{-1}(s)\lvert < q_{s}$, then $ GS^{\ell}_{i}(P^{s}_{i},P^{\ast}_{-i},\succ,q)=s$. If $ \lvert \mu^{-1}(s) \lvert =q_{s}$, then by Case 2 above, student $ i$ has higher priority than some student in $\mu^{-1}(s) $. Therefore, $ GS^{\ell}_{i}(P^{s}_{i},P^{\ast}_{-i},\succ,q)=s$. Therefore the admission to school $ s$ is not strategy-proof to student $ i$ via $ GS^{\ell}$.
\par
We provide an environment where the admission to some school is strategy-proof to a student via $ GS^{\ell}$ but not via $ GS^{k}$. Fix an environment $ (\succ,q)$ where schools have the same priority and where each has one seat. Then $ GS^{\ell}$ is strategy-proof to the top $ \ell$ students under the common priority but not to the $ (\ell+1)$'s priority student. Since $ k\geq \ell+1$, then $ GS^{k}$ is strategy-proof to the $ (\ell+1)$'s priority student.

\end{proof}

\begin{proof}[Proof of Corollary \ref{cor1}]
Fix an environment $ (\succ,q)$ and suppose that the admission to school $ s$ is not strategy-proof to student $ i$ via $ GS^{k}$. By Theorem \ref{pro3} the admission to school $ s$ is also not strategy-proof to student $ i$ via $ GS^{\ell}$. By Theorem \ref{pro2} the admission to school $ s$ is also not strategy-proof to student $ i$ via $ FPF^{\ell}$. \par
It remains to provide an environment where the converse is not true. We provided such an environment in the proof of Theorem \ref{pro2}: there the admission to some school $ s$ is strategy-proof to student $ i$ via $ GS^{\ell}$ but not via $ FPF^{\ell}$. By Theorem \ref{pro3} the admission to school $ s$ is strategy-proof to student $ i$ via $ GS^{k}$.
\end{proof}

\textit{\underline{Theorem \ref{Chinese}}: Let $ e > e'$ and suppose that there is at least $ e$ schools. Then $ Ch^{(e)}$ is more immune to strategic admissions than $ Ch^{(e')}$.}

\begin{proof}[Proof of Theorem \ref{Chinese}] 
First, we collect some basic results that are proven in \cite{chenandkesten2017a} needed to prove Theorem \ref{Chinese}.
\begin{lemma}\citep{chenandkesten2017a}.\label{Chinese-lemma} Let $ e$ be given. Let $ P$ be a preference profile, $ i$ a student, $ s'$ a school and $ P_{i}^{s'}$ a preference relation where $ i$ has ranked school $ s'$ first.
\begin{itemize}
\item[(i)] Suppose that student $ i$ is matched to school $ s'$ under $ Ch^{(e)}(P,\succ,q)$. Then he is also matched to school $ s'$ under $ Ch^{(e)}(P^{s'}_{i},P_{-i},\succ,q)$.
\item[(ii)] Suppose that student $ i$ prefers school $ s'$ to his matching under $Ch^{(e)}(P,\succ,q) $ and has ranked it among his top $ e$ schools under $ P$. Then he cannot obtain a seat at school $ s'$ by misrepresenting his preferences.
\end{itemize}
\end{lemma}
 We prove the theorem by contraposition. Fix an environment $ (\succ,q)$ and suppose that the admission to school $ s$ is not strategy-proof to student $ i$ via $ Ch^{(e)}$. Then there is a preference profile $ P$ and a preference relation $ P'_{i}$ such that
\begin{equation}
s=Ch^{(e)}_{i}(P'_{i},P_{-i},\succ,q)\mathrel{P_{i}} Ch^{(e)}_{i}(P,\succ,q). \label{ch10}
\end{equation}
Let $ P_{i}^{s}$ be a preference relation where student $ i$ has ranked school $ s$ first. By Lemma \ref{Chinese-lemma} (i), $ Ch^{(e)}_{i}(P^{s}_{i},P_{-i},\succ,q)=s$. Then student $ i$ is matched in the first round of the mechanism. Thus,
\begin{equation}
GS_{i}(P^{s}_{i},P_{-i}^{e},\succ,q)=s. \label{ch}
\end{equation}
 Since $ s \mathrel{P_{i}} Ch^{(e)}_{i}(P,\succ,q)$, then student $ i$ has been rejected by school $ s$ at some round. Hence all the seats of school $ s$ have been assigned under $ \mu=Ch^{(e)}(P,\succ,q)$. That is $ \lvert \mu^{-1}(s)\lvert =q_{s}$. Let $ N=\mu^{-1}(s)$ denote the set of students who are matched to school $ s$ under $ Ch^{(e)}(P,\succ,q)$. 
 
 We now prove that the admission to school $s $ is not strategy-proof to student $ i$ via $ Ch^{(e')}$. By Lemma \ref{Chinese-lemma} (ii), student $ i$ did not rank school $ s$ among his top $ e$ schools under $P $. By equation \ref{ch10}, if $ \mu(i)$ is a school, then it is ranked lower that the position $ e$ under $ P_{i}$. We consider two cases:
\par \textbf{Case 1:} At least one student in $ N$ has lower priority than student $ i$ under $ \succ_{s}$. Since $ e'< e$, student $ i$ has ranked more than $ e'$ schools above school $ s$ under $ P_{i}$. For each $ \ell=1,\hdots, e'$, let $ s_{\ell}$ be the $ \ell'$th ranked school under $ P_{i}$. Let $ P^{\ast}$ denote the following preference profile:

\begin{center}
\begin{tabular}{c c}
$ P^{\ast}_{i}$ & $ P^{\ast}_{j\neq i}$\\\hline
$ s_{1}$ & $ \mu (j)$\\
$ \vdots$ & $ \emptyset$\\
$ s_{e'}$ & \\
$ s$ &\\
$ \mu(i)$\\
$ \emptyset$
\end{tabular}
\end{center}

Note that student $ i$ is not matched in the first round $ Ch^{(e)}$. Thus $ GS_{i}(P^{e},\succ,q)=\emptyset$. Then for each $ \ell=1,\hdots,e'$, each student matched to school $ s_{\ell}$ in $ \mu$ has higher priority than student $ i$ under $ \succ_{s_{\ell}}$.\footnote{This is because for each $ \ell=1,\hdots,e'$ and each student $ j$ such that $ \mu(j)=s_{\ell}$, $ s_{\ell}=GS_{j}(P^{e},\succ,q)$ and $ s_{\ell}\mathrel{P^{e}_{i}\mu(i)}$.} Furthermore, under $ P^{\ast}$, each student in $ N$ has ranked school $ s$ first and student $ i$ did not rank it among the top $ e'$. Therefore,

\begin{equation*}
     Ch^{(e')}(P,\succ,q)=\mu. 
\end{equation*}

Since at least one student in $ N$ has lower priroity than student $ i$ under $ \succ_{s}$ we have,

\begin{equation*}
     Ch^{(e')}(P^{s}_{i},P^{\ast}_{-i},\succ,q)=s,
\end{equation*}

\noindent proving that the admission to school $ s$ is not strategy-proof to student $ i$ via $ Ch^{(e')}$.
\par \textbf{Case 2:} Every student in $ N$ has higher priority than student $ i$ under $ \succ_{s}$. We claim that \\ 

\textit{Claim:} at least one student in $ N$ has ranked school $ s$ below position $ e$ under $ P$.

\begin{proof}[Proof of the claim]
Suppose, to the contrary, that every student in $ N$ has ranked school $ s$ among the top $ e$ schools under $ P$. Let $ \eta=GS(P^{e},\succ,q)$. Student $ i$ is not matched to his top $ e$ schools under $ Ch^{(e)}(P,\succ,q)$. Therefore, $ \eta(i)=\emptyset$. Every student in $ N$ is matched to school $ s$ under $ \mu$ and has ranked it among the top $ e$ schools under $ P$. Then for each $ j\in N$, $ \eta(j)=s$. Because every student in $ N$ has higher priority than student $ i$ under $ \succ_{s}$, $ \eta$ is also stable under $ (P^{s}_{i},P^{e}_{-i},\succ,q)$. By Roth (1984) the set of matched students is the same at all stable matchings. Hence $ GS_{i}(P^{s}_{i},P^{e}_{-i},\succ,q)=\emptyset$. This contracts equation \ref{ch}.
\end{proof}
Since there is at least one student in $ N$ who ranked school $ s$ below position $ e$ under $ P$ and that $ e'< e$, there is at least one student in $ N$ who ranked school $ s$ below position $ e'$ under $ P$. Let $ j$ be one such student. For each $ \ell=1,\hdots,e'$, let $ s^{i}_{\ell}$ and $ s^{j}_{\ell}$ denote the $ \ell'$th ranked school of student $ i$ and $ j$, respectively, under $ P$. Let $ P^{\ast}$ be the following profile.
\begin{center}
\begin{tabular}{c c c}
$ P^{\ast}_{i}$ & $ P^{\ast}_{j}$ & $ P^{\ast}_{k\neq i,j}$\\\hline
$ s_{1}^{i}$ & $ s_{1}^{j}$ & $ \mu (k)$ \\
$ \vdots$ & $ \vdots$ & $ \emptyset$\\
$ s_{e'}^{i}$ & $s_{e'}^{j} $ & \\
$ s$ & $s $ & \\
$ \mu(i)$ & $ \emptyset$ &\\
$ \emptyset$ & $ $ &
\end{tabular}
\end{center}
All seats of each of the schools $ s^{j}_{1}, \hdots, s^{j}_{e'}$ are assigned in the first round of $ Ch^{(e)}(P,\succ,q)$. Since student $ i$ is matched (if any) in a round later than the first round of $ Ch^{(e)}(P,\succ,q)$, $ \mu(i) $ is not one of the schools $ s^{j}_{1}\hdots, s^{j}_{e'}$. Let $\ell=1\hdots,e' $. Because student $ j$ has ranked school $ s^{j}_{\ell}$ among the top $ e$ schools under $ P$ and has been rejected, all its seats have been assigned at the first round of $ Ch^{(e)}(P,\succ,q)$ to students who have higher priority than him under $ \succ_{s_{\ell}^{j}}$. Thus each student in $\mu^{-1}( s^{j}_{\ell})$ has higher priority than student $ j$ under $ \succ_{s_{\ell}^{j}}$. Then $ Ch^{(e')}(P^{\ast},\succ,q)$ is completed in two rounds, and because student $ j$ has higher priority than student $ i$ under $ \succ_{s}$, we have
\begin{equation*}
Ch^{(e')}(P^{\ast},\succ,q)=\mu.
\end{equation*}
Now under $(P^{s}_{i},P^{\ast}_{-i})$, there is $q_{s}$ students (including student $ i$) who have ranked school $ s$ among the top $ e'$ schools. Therefore $ Ch^{(e')}_{i}(P^{s}_{i},P^{\ast}_{-i},\succ,q)=s$ and 

\begin{equation*}
    s= Ch^{(e')}_{i}(P^{s}_{i},P^{\ast}_{-i},\succ,q) \mathrel{P^{\ast}_{i}} Ch^{(e')}_{i}(P^{\ast},\succ,q)=\mu(i),
\end{equation*}

\noindent proving that the admission to school $ s$ is not strategy-proof to student $ i$ via $ Ch^{(e')}$.

\end{proof}

\textit{\underline{Proposition \ref{lemma1}}: Let $ k\geq 1$ and $ (\succ,q)$ an environment where schools have a common priority. Let the capacities of the schools be increasingly ordered $ q_{1}\leq q_{2} \leq\hdots \leq  q_{\lvert S\lvert}$ and $ \hat{q}=q_{1}+\hdots + q_{k}$. Then, the mechanism $ SD^{k}$ is strategy-proof for the $\hat{q}$-highest priority students.}

\begin{proof}[Proof of Proposition \ref{lemma1}] The mechanism $ SD$ is strategy-proof. Let $ P$ be a preference profile and suppose that student $ i$ is matched under $ SD^{k}(P,\succ,q)$ or has ranked less than or $ k$ acceptable schools. By Lemma \ref{lemma}, student $ i$ cannot manipulate $ SD^{k}$ at $ P$. Let $ i$ be one of the $ \hat{q}$-highest priority students. We show that he never misses one of his $ k$ most preferred schools whenever he ranks at least $ k$ acceptable schools. This will complete the proof. 

Suppose, to the contrary, that student $ i$ ranks at least $ k$ acceptable schools and ends up unmatched under $ SD^{k}(P,\succ,q)$. Let $ S'$ denote the set of his $ k$ most preferred schools. Then, at his turn, all the seats of the schools in $ S'$ have been selected. Then at least $ q'=\sum_{s\in S'}q_{s}$ students have moved before student $ i$. Then student $ i$ is not one of the $ q'$-highest priority students. This contradicts the fact that student $ i$ is one of the $ \hat{q}$-highest priority students because $ \hat{q}\leq q'$.

\end{proof}

\textit{\underline{Theorem \ref{thm4}}: Let $ k>\ell$ and suppose that there are at least $ k$ schools and at least one first-preference-first school. Then
\begin{itemize}
    \item $ GS^{k}$ is strongly more immune to strategic admissions than $ GS^{\ell}$,
    \item $ GS^{k}$ is strongly more immune to strategic admissions than $ FPF^{k}$.
\end{itemize}}

\begin{proof}[Proof of Theorem \ref{thm4}] We prove the theorem by contraposition. Suppose that there are at least $ k$ schools and at least one first-preference-first school. Let an environment $ (\succ,q)$ be given and suppose that the admission to school $ s$ is not strategy-proof to student $ i$ via $ GS^{k}$ in equilibrium. There is a preference profile $ P$ and a preference relation $ P'_{i}$ such that
\begin{itemize}
    \item $ (P'_{i},P_{-i})$ is a Nash equilibrium of the $ (P,GS^{k})$ and 
    \item  $s=GS^{k}_{i}(P'_{i},P_{-i},\succ,q) \mathrel{P_{i}} GS^{k}_{i}(P,\succ,q)$.

\end{itemize}
 We show that the admission to school $ s$ is not strategy-proof to student $ i$ via $ FPF^{k}$ and $GS^{\ell} $ in equilibrium. The difference with the proof of Theorem \ref{pro2} and Theorem \ref{pro3} is that we further assumed that $ (P'_{i},P_{-i})$ is a Nash equilibrium of the $ (P,GS^{k})$. For each of the preference profiles that we constructed in equations \ref{eqFP} and \ref{eqGS}, we have
 \begin{equation*}
    s=GS^{k}_{i}(P^{s}_{i},P^{\ast}_{-i},\succ,q) \mathrel{P^{\ast}_{i}} GS^{k}_{i}(P^{\ast},\succ,q),
 \end{equation*}
 where student $ i$ has ranked school $ s$ first under $ P^{s}_{i}$. Note that there is a student in $\mu^{-1}(s)$ who has lower priority than student $ i$ under $ \succ_{s}$. Let $ j$ be the lowest priority students among them. Now, under $ FPF^{k}(P^{s}_{i},P^{\ast},\succ,q)$ student $ j$ is unmatched, student $ i$ is matched to school $ s$ and each of the remaining students is matched to their first choice school. The strategy $ (P^{s}_{i},P^{\ast}_{-i})$ is a Nash equilibrium of $ (P,FPF^{k})$. Indeed, student $ i$ cannot get a seat at a school $ s'$ that he prefers to $ s$ because each student in $ \mu^{-1}(s')$ has ranked $ s'$ first and has higher priority than him under $ \succ_{s'}$. Similarly, each student matched to school $ s$ under $ FPF^{k}(P^{s}_{i},P^{\ast}_{-i},\succ,q)$ has higher priority than student $ j$ under $ \succ_{s} $ and has ranked it first under $ (P^{s}_{i},P^{\ast}_{-i})$. Thus student $ j$ cannot be matched to school $ s$ by reporting a preference relation other than $ P^{\ast}_{j}$. The prove that the admission to school $ s$ is not strategy-proof to student $ i$ via $ FPF^{k}$ in equilibrium. 
 
 The argument can be used to prove that $ (P^{s}_{i},P^{\ast}_{-i})$ is a Nash equilibrium of the game $ (P,GS^{\ell})$ where $ P^{\ast}$ is the preference profile in equation \ref{eqGS}.

We provide an environment where the admission to some school is strategy-proof to some student via $ GS^{k}$ in equilibrium but not via $ FPF^{k}$. Let $ (\succ,q)$ be an environment where the schools have a common priority and where each school has one seat. By assumption there are at least $ k\geq 2$ schools and students. Let students be ordered from $ 1$, the highest priority student, to $ \lvert I\lvert$, the lowest priority student. By Proposition \ref{lemma1}, $ GS^{k}=SD^{k}$ is strategy-proof for student $ 2$. Without loss of generality, let us assume that school $ s_{2}$ is a first-preference-first school. Let $ P$ be the following preference profile:

\begin{center}
    \begin{tabular}{c c c c}
        $ P_{1}$ & $ P_{2}$ & $ P_{3}$ & $ P_{-\{1,2,3\}}$\\\hline
        $ s_{1}$ & $ s_{1}$ & $ s_{2}$ & $\emptyset$\\
        $ \emptyset$ & $ s_{2}$ & $ \emptyset$ & \\
         $ $ & $ \emptyset$ & $ $ & 
        
    \end{tabular}
\end{center}

Since $ k\geq 2$, $ FPF_{2}^{k}(P,\succ,q)=\emptyset$. Let $ P^{s_{2}}_{2}$ be a preference relation where student $ 2$ ranks school $ s_{2}$ first. Clearly, $ (P^{s_{2}}_{2},P_{-2})$ is a Nash equilibrium of the game $ (P,FPF^{k})$ and 
\begin{equation*}
    s_{2}=FPF^{k}_{2}(P^{s_{2}}_{2},P_{-2},\succ,q) \mathrel{P_{2}} FPF^{k}_{2}(P,\succ,q).
\end{equation*}
Therefore, the admission to school $ s_{2}$ is not strategy-proof to student $ 2$ via $ FPF^{k}$ in equilibrium. 

We consider the same environment to show that the admission to some school is strategy-proof to some student via $ GS^{k}$ in equilibrium but not via $ GS^{\ell}$. Since $ k> \ell$, $ k\geq \ell+1$ and by the Proposition \ref{lemma1}, $ GS^{k}=SD^{k}$ is strategy-proof to the student $ \ell+1$. Let $ P$ be a preference profile such that for each student $ i$, $ P_{i}$ is specified as follows (recall that there is at least $ k$ schools).

\begin{center}
    \begin{tabular}{c }
        $ P_{i}$ \\\hline
        $ s_{1}$ \\
        $ \vdots$\\
        $ s_{\ell}$\\
        $ s_{\ell+1}$\\
        $\emptyset$
    \end{tabular}
\end{center}
Then $ GS^{\ell}_{\ell+1}(P,\succ,q)=\emptyset$. Let $ P^{s_{\ell+1}}_{\ell+1}$ be a preference relation where student $ \ell+1$ ranked school $ s_{\ell+1}$ first. Clearly, $ (P^{s_{\ell+1}}_{\ell+1},P_{-(\ell+1)})$ is a Nash equilibrium of the game $ (P,GS^{\ell})$, and 
\begin{equation*}
    s_{\ell+1}=GS^{\ell}_{\ell+1}(P^{s_{\ell}}_{\ell+1},P_{-(\ell+1)},\succ,q) \mathrel{P_{\ell+1}} GS^{\ell}_{\ell+1}(P,\succ,q).
\end{equation*}

Therefore, the admission to school $ s_{\ell+1}$ is not strategy-proof to student  student $\ell+1 $ via $ GS^{\ell}$ in equilibrium.
\end{proof}

\end{document}